\documentclass[journal,onecolumn]{IEEEtran}
\usepackage{amssymb}
\usepackage{amsmath}
\usepackage{longtable}
\newtheorem{theorem}{Theorem}[section]
\newtheorem{lemma}[theorem]{Lemma}

\newtheorem{corollary}[theorem]{Corollary}

\newenvironment{proof}[1][Proof]{\begin{trivlist}
\item[\hskip \labelsep {\bfseries #1}]}{\end{trivlist}}
\newenvironment{definition}[1][Definition]{\begin{trivlist}
\item[\hskip \labelsep {\bfseries #1}]}{\end{trivlist}}

\newenvironment{remark}[1]
[Remark]{\begin{trivlist}
\item[\hskip \labelsep {\bfseries #1}]}{\end{trivlist}}

\makeatletter

\newcommand{\Rmnum}[1]{\expandafter\@slowromancap\romannumeral #1@}
\makeatother

\newcommand{\fq}{{\mathbb F}_{2^{2m}}}
\newcommand{\fqm}{{\mathbb F}_{2^{2m}}^{\times}}

\newcommand{\Tr}{{\rm {Tr}}}

\newcommand{\ftau}{{\mathbb F}_{2}^{\tau}}

\ifCLASSINFOpdf

\else

\fi

\begin{document}
%
\title{Constructing   bent functions and  bent idempotents of any possible algebraic degrees}

\author{Chunming~Tang,
 Yanfeng~Qi, Zhengchun Zhou, Cuiling Fan
\thanks{C. Tang is with School of Mathematics and Information, China West Normal University, Sichuan Nanchong, 637002, China. e-mail: tangchunmingmath@163.com
}

\thanks{Y. Qi is with School of Science, Hangzhou Dianzi University, Hangzhou, Zhejiang, 310018, China.
e-mail: qiyanfeng07@163.com.
}
\thanks{
Z. Zhou is with the School of Mathematics, Southwest Jiaotong University, Chengdu, 610031, China. e-mail:
zzc@swjtu.edu.cn.}
\thanks{
C. Fan is with the School of Mathematics, Southwest Jiaotong University, Chengdu, 610031, China. e-mail:
fcl@swjtu.edu.cn.}
}


\maketitle

\begin{abstract}
Bent functions as optimal combinatorial objects are difficult to characterize and construct.
In the literature,  bent idempotents are a special class of bent functions and few constructions have been presented, which are restricted by the degree of  finite fields  and have algebraic degree no more than 4. In this paper, several new infinite families of bent functions are obtained by adding the the algebraic combination of linear functions to some known bent functions and their duals are calculated. These bent functions contain some previous work on infinite families of bent functions by Mesnager \cite{M2014} and Xu et al. \cite{XCX2015}. Further,
infinite families of bent
idempotents
of any possible algebraic degree
are constructed from  any
quadratic bent idempotent. To our knowledge,
it is the first univariate representation  construction of
infinite families  of bent idempotents over  $\mathbb{F}_{2^{2m}}$ of
algebraic degree between 2 and $m$, which solves the  open problem on bent idempotents proposed by  Carlet \cite{C2014}. And an  infinite
family  of anti-self-dual bent functions
are obtained. The sum of three  anti-self-dual bent functions in such a family is also anti-self-dual bent and belongs to this family. This solves the open problem proposed by Mesnager  \cite{M2014}.
\end{abstract}


%
\IEEEpeerreviewmaketitle

\section{Introduction}
Boolean bent functions introduced by Rothaus \cite{R1976} in 1976 are an interesting combinatorial object with the maximum
Hamming distance to the set of all affine functions. Such functions have been extensively studied because of
their important applications in cryptograph (stream ciphers \cite{C2010}), sequences \cite{OSW1982},
graph theory \cite{PTF2010}, coding theory ( Reed-Muller codes \cite{CHLL1997}, two-weight and three-weight linear codes \cite{CK1986,Ding2015}), and association schemes \cite{PTFL2011}.
A complete classification of bent functions is still elusive. Further, not only
their characterization, but also
their generation are challenging problems.
Much work on bent functions are devoted to
the construction of bent functions
\cite{CCK2008,C1994,C1996,C2010,CM2011,CG2008,
CK2008,D1974,DD2004,DLCCFG2006,G1968,
L2006,LK2006,LHTK2013,M1973,M2009,M2010,M2011-1,
M2011-2,M2014,MF2013,YG2006}.

Idempotents introduced by Filiol and Fontaine in \cite{F1999,FF1998} are Boolean functions over
$\mathbb{F}_{2^n}$ such that for any
$x\in \mathbb{F}_{2^n}$,
$f(x)=f(x^2)$. Rotation symmetric Boolean functions, which was also introduced by Filiol and Fontaine under the name of idempotent functions and studied by Pieprzyk and Qu \cite{PQ1999}, are invariant under circular translation of indices. Due to less space to be stored and
allowing faster computation of the Walsh transform, they are of great interest.
They can be obtained from
idempotents (and vice versa) through the choice of a normal basis of $\mathbb{F}_{2^n}$. Characterizing and
constructing idempotent bent functions and rotation symmetric bent functions are difficult and have theoretical and practical interest.   In the literature,
few constructions of bent idempotents have been presented, which are restricted by
the degree of finite fields and have algebraic degree no more than 4. Carlet \cite{C2014}
introduced an open problem: how to construct classes of
bent idempotents over $\mathbb{F}_{2^{2m}}$ of algebraic degree between 5 and $m$.

Bent functions always occurs in pairs.
In fact, the dual of  a given bent function
is also a bent function. Generally, the calculation of the dual of a given bent function is difficult. As a subclass of bent functions, self-dual bent functions (resp. anti-self-dual bent functions )
coincide with their dual (resp. have a
very simple way with their dual). They are
interesting and can be used to construct new bent functions. Mesnager \cite{M2014} proposed
an open problem: find three anti-self-dual
bent functions $f_1$, $f_2$, and
$f_3$ over $\mathbb{F}_{2^{2m}}$ such that
$f_1+f_2+f_3$ is anti-self-dual bent.

In this paper, we consider bent functions of the form
$$
f(x)=g(x)+F(\Tr_1^n(u_1x),\Tr_1^n(u_2x)
,\cdots, \Tr_1^n(u_\tau x)),
$$
where $g(x)$ is some known bent function,
$n=2m$,
$\tau$ is a positive integer such that
$1\leq \tau \leq m$,
$u_1,u_2,\cdots, u_\tau\in \fqm$,
 and
$F(X_1,X_2,\cdots,X_\tau)
\in \mathbb{F}_2[X_1,X_2,\cdots,X_\tau]
$. Mesnager \cite{M2014} studied the case
that $F(X_1,X_2)=
X_1X_2$. Xu et al. \cite{XCX2015} studied
the case $F(X_1,X_2,X_3)=
X_1X_2X_3$. We consider the general case.
First, from some known bent functions, we present several new infinite
families of bent functions by calculating
their duals. Second, from any quadratic idempotent bent function, we get  infinite families of bent idempotents over $\fq$ of  algebraic degree between 2 and $m$. To our knowledge, it is the first univariate representation  construction of
infinite families  of bent idempotents over  $\mathbb{F}_{2^{2m}}$ of
algebraic degree between 2 and $m$, which solves
the open problem proposed by Carlet \cite{C2014}.
Third, we present
an infinite family of anti-self-dual bent functions over $\fq$ of algebraic degree
between 2 and $m-1$. If
$f_1$, $f_2$, and $f_3$ are anti-self-dual bent functions of this family, then
$f_1+f_2+f_3$ is also anti-self-dual bent and belongs to this family. Hence, we solve the
open problem proposed by Mesnager \cite{M2014}.

The rest of the paper is organized as follows: Section \Rmnum{2} introduces some basic
notations and results on Boolean bent functions, bent idempotents, and
self-dual bent functions (anti-self-dual bent
functions). Section \Rmnum{3}, we present our concrete construction of several new infinite  families of bent functions. Concretely, in Section
\ref{subA}, a new infinite family of bent functions from Kasami bent functions
is obtained, which gives
an infinite family of bent idempotents
and an infinite family of anti-self-dual
bent functions.
We obtain an  infinite family of    bent  functions and bent idempotents  from any quadratic idempotent bent function  in Section \ref{quadratic},
a new infinite family of bent functions from Gold-like monomial functions in Section \ref{subB}, a  new infinite family of bent functions from Niho exponents in Section \ref{subC}, and a new infinite family of bent functions from the
Maiorana-Mcfarland class of bent functions in Section \ref{subD}.
Section \Rmnum{4} makes a conclusion.

\section{Preliminaries}

Let $n=2m$ be a positive integer,
$\fq$ denote the finite field with $2^{2m}$
elements, and $\fqm=\fq\backslash \{0\}$.
For any positive integer $k|n$,
the trace function from $\mathbb{F}_{2^n}$
to $\mathbb{F}_{2^k}$, denoted by
$\Tr_k^n$, is the mapping defined as :
$$
\Tr_k^n(x)=x+x^{2^k}+
 x^{2^{2k}}+\cdots+ x^{2^{n-k}}.
$$
When $k=1$, $\Tr_1^n(x)
=\sum_{i=0}^{n-1}x^{2^i}$ is called the absolute trace function. The trace function satisfies that $
\Tr_1^n(x)=\Tr_1^n(x^2)$ and $\Tr_1^n(x)=\Tr_1^k(\Tr_k^n(x))$
for any integer $k|n$.
Let $\mathbb{Z}$ be the integer set and
$\mathbb{F}_2[X_1,X_2,\cdots,X_\tau]$
the multivariate polynomial ring
with variables $X_1,X_2,\cdots,X_\tau$ for any positive integer $\tau$.

A Boolean function $f(x)$ defined on $\mathbb{F}_{2^n}$ is a mapping from $\mathbb{F}_{2^n}$ to $\mathbb{F}_2$. Every nonzero $f(x)$ has a unique univariate representation
of the form
$$
f(x)=\sum_{j=0}^{2^n-1}a_{j}x^j.
$$
where $a_{0}, a_{2^n-1}\in\mathbb{F}_{2}$
 and $a_{2j}=a_{j}^2$, and
$j \in \mathbb{Z}/(2^n-1)
\mathbb{Z}$.
The univariate representation can be written as a trace representation of the form
$$
f(x)=\sum_{j\in \Gamma_{n}}
\mathrm{Tr}_{1}^{o(j)}(a_{j}x^j)+
a_{2^n-1}x^{2^n-1}
$$
where
\begin{itemize}
\item $\Gamma_n$ is the set of integers obtained by choosing one element in each cyclotomic class of 2 modulo
    $2^n-1$ ($j$ is often chosen as the smallest element in its cyclotomic class, called the coset leader of the class);
\item $o(j)$ is the size of the cyclotomic coset of 2 modulo $2^n-1$ containing $j$;
\item $a_{j}\in\mathbb{F}_{2^{o(j)}}$;
\item $a_{2^n-1}=0$ or 1.
\end{itemize}
We also have $f(x)=\Tr_1^n(P(x))$ for some polynomial $P(x)$, which is not a unique
representation.
\begin{definition}
Let $f(x)$ be a Boolean function defined on
$\mathbb{F}_{2^n}$. $f(x)$ is called an
idempotent if
$$
f(x)=f(x^2), \forall x\in \mathbb{F}_{2^n}.
$$
\end{definition}

Every
$x\in \mathbb{F}_{2^n}$ has a unique representation $x=X_1\zeta_1+
X_2\zeta_2+\cdots+X_n\zeta_n$, where
$X_i\in \mathbb{F}_2$ and
$\zeta_1,\zeta_2,\cdots,\zeta_n$ is a
basis of $\mathbb{F}_{2^n}$ over
$\mathbb{F}_2$. Hence, $x$ is denoted by
a vector $(X_1,X_2,\cdots,X_n)$.
A Boolean function $f(x)$ can also be written as $f(X_1,X_2,\cdots,X_n)
:\mathbb{F}_2^n\longmapsto \mathbb{F}_2$.
A Boolean function defined over $\mathbb{F}_{2}^n$ is often
represented by the algebraic normal form (ANF):
\begin{equation}\label{anf}
f(X_{1}, \cdots, X_{n})=
\sum_{I\subseteq \{1,\cdots,n\}}
a_{I}(\prod_{i\in I}X_{i}),~~~ a_{I}\in \mathbb{F}_{2}.
\end{equation}
A polynomial  in $\mathbb{F}_2[X_1,X_2,\cdots,X_n]$
with the form in Equation (\ref{anf}) is
called a reduced polynomial.

The algebraic degree of a Boolean function
$f$ is the global degree of its algebraic
normal form. And the algebraic degree of $f$
can also be expressed as
$max\{wt_2(j):a_j\neq 0\}$, where
$wt_2(j)$ is the Hamming weight of
the binary expression of $j$ and
$a_j$ is in the trace representation or
univariate representation  of $f$. The following lemma is used to determine the algebraic degree of some Boolean functions.
\begin{lemma}\label{degree}
Let $\tau$ be a positive integer such that
$1\leq \tau \leq m$, $u_1,u_2,\cdots, u_\tau\in \mathbb{F}_q^{\times}$
be linearly independent over $\mathbb{F}_2$, and  $F(X_1,X_2,\cdots,X_\tau)$ be a reduced polynomial function in $\mathbb{F}_2[X_1,X_2,\cdots,X_\tau]$ of algebraic degree $d$.
Then the function $$
F(\Tr_1^n(u_1x),\Tr_1^n(u_2x)
,\cdots, \Tr_1^n(u_\tau x))
$$
 has algebraic degree $d$.
\end{lemma}
\begin{proof}
Since $u_1,u_2,\cdots, u_\tau$ are linearly independent, there exist $u_{\tau+1},
\cdots,u_n$ such that $u_1,u_2,\cdots, u_n$ form  a basis of $\mathbb{F}_q$ over $\mathbb{F}_2$. Thus, for any $u\in \mathbb{F}_q$, there exists $c_i\in \mathbb{F}_2$ such that
$u=\sum_{i=1}^nc_iu_i$.
Define a mapping
\begin{align*}
\sigma: \mathbb{F}_{2^n}&\longrightarrow \mathbb{F}_2^n\\
x&\longmapsto (X_1,X_2,\cdots,X_n),
\end{align*}
where $X_i=\Tr_1^n(u_ix)$. If there exists
$x_0\in \mathbb{F}_q$ such that
$$
\Tr_1^n(u_1x_0)=\Tr_1^n(u_2x_0)
=\cdots=\Tr_1^n(u_\tau x_0)=0,
$$
then for any $u\in \mathbb{F}_q$, $\Tr_1^n(ux_0)=\sum_{i=1}^{n}c_i
\Tr_1^n(u_ix_0)
=0$. Hence, $x_0=0$.  The defined
mapping $\sigma$ is nondegenerated.
$F(\Tr_1^n(u_1x),\Tr_1^n(u_2x)
,\cdots, \Tr_1^n(u_\tau x))$ can be seen as
a  function $F(X_1,X_2
,\cdots, X_\tau)$ over $\mathbb{F}_2^{n}$.
The algebraic degree of a Boolean function
under a nondegenerate mapping
keeps unchanged. Hence, $F(\Tr_1^n(u_1x),\Tr_1^n(u_2x)
,\cdots, \Tr_1^n(u_\tau x))$  has algebraic degree $d$.
\end{proof}
\begin{definition}
A Boolean function  or
a multivariate polynomial $f(X_1,X_2,\cdots,
X_n)$ is  rotation symmetric   if it is invariant under cyclic shift:
$$
f(X_n,X_1,X_2,\cdots,
X_{n-1})=f(X_1,X_2,\cdots,
X_n)
$$
\end{definition}
For a Boolean function $f(x)$ and a normal basis $\{u,u^2,\cdots,
u^{2^{n-1}}\}$  of $\mathbb{F}_{2^n}$, the
multivariate representation of $f(x)$:
$$
f(X_1,X_2,\cdots,X_n)=f(\sum_{i=1}^{n}
X_iu^{2^{i-1}})
$$
is rotation symmetric if and only if
$f(x)$ is an idempotent.

The bivariate representation is
based on the identification
$\mathbb{F}_{2^{2m}}
\approx \mathbb{F}_{2^m}
\times \mathbb{F}_{2^m}$ and has the form
$$
f(x)=\sum_{0\leq i,j\leq
2^m-1}a_{i,j}x^iy^j, ~~~a_{i,j}\in
\mathbb{F}_{2^{m}}
$$
The simple representation $f(x,y)=
\Tr_1^m(P(x,y))$ exists but not unique, where $P(x,y)$ ia a polynomial
over $\mathbb{F}_{2^m}$.

The Walsh transform of a Boolean function calculates the correlations between the function and linear Boolean functions.
For the univariate trace representation
over
$\mathbb{F}_{2^n}$, the Walsh transform of
$f$ is
$$
\mathcal{W}_f(\beta)=
\sum_{x\in \mathbb{F}_{2^n}}
(-1)^{f(x)+\Tr_1^n(\beta x)}, ~~~
\beta \in \mathbb{F}_{2^n}.
$$
For  the bivariate trace representation over $\mathbb{F}_{2^m}\times
\mathbb{F}_{2^m}$, the Walsh transform of $f$ is
$$
\mathcal{W}_f(\beta_1,\beta_2)=
\sum_{x,y\in \mathbb{F}_{2^m}}
(-1)^{f(x,y)+\Tr_1^n(\beta_1 x+
\beta_2 y)}.
$$
For the multivariate representation over
$\mathbb{F}_{2}^n$, the Walsh transform of
$f$ is
$$
\mathcal{W}_f(y_1,y_2,\cdots, y_n)=
\sum_{(x_1,x_2,\cdots,x_n)\in \mathbb{F}_{2}^n}
(-1)^{f(x_1,x_2,\cdots,x_n)+
\sum_{i=1}^{n}x_iy_i},~~~~y_i\in \mathbb{F}_2.
$$
The inverse Walsh transform of
$f(x_1,x_2,\cdots,x_n)$ is
$$
(-1)^{f(x_1,x_2,\cdots,x_n)}
=\frac{1}{2^n}\sum_{(y_1,y_2,\cdots,y_n)
\in \mathbb{F}_2^n}\mathcal{W}
(y_1,y_2,\cdots,y_n)(-1)^{x_1y_1+
x_2y_2+\cdots+x_ny_n}.
$$
Then $(-1)^{f(x_1,x_2,\cdots,x_n)}$
can be expressed by the linear combination of $(-1)^{x_1y_1+
x_2y_2+\cdots+x_ny_n}$. This representation
is unique and is vital in this paper. And we rewrite this in the following lemma.
\begin{lemma}\label{lem2.2}
Let $\tau$ be a positive  integer and
$F(X_1,X_2,\cdots,X_\tau)\in
\mathbb{F}_2[X_1,X_2,\cdots,X_\tau]$.
There exists a  unique set of
$c_{\mathbf{w}}\in \mathbb{C}$ satisfying
that
\begin{equation*}
(-1)^{F(X_1,X_2,\cdots,X_\tau)}
=\sum_{\mathbf{w}
\in \mathbb{F}_2^{\tau}}
{c}_{\mathbf{w}}
(-1)^{\sum_{i=1}^{\tau}w_iX_i},
\end{equation*}
where $\mathbf{w}=(w_1,w_2,\cdots,w_\tau)
\in \ftau$.
\end{lemma}
\begin{definition}
A Boolean function $f$ is called a bent
function if $|\mathcal{W}_f(\beta)|
=2^{n/2}$ for any $\beta \in \mathbb{F}_{2^n}$. If $f$ is an idempotent,
then $f$ is called a bent idempotent or
an idempotent bent function.
\end{definition}
Bent functions exist only for an even number of variables. For a bent function defined over $\mathbb{F}_{2^n}$, its dual
is the Boolean function $\widetilde{f}$
such that $\mathcal{W}_f(\beta)
=2^{n/2}(-1)^{\widetilde{f}(\beta)}$.
The dual of a bent function is also bent.
Thus, Boolean bent functions occur in pairs.
Determining the dual of a bent function
is difficult.
A bent function is called  self-dual
(resp. anti-self-dual) if
$\widetilde{f}=f$ (resp.
$\widetilde{f}=f+1$).

\section{Several new infinite families of bent functions and bent idempotents}
In this section, we generalize the construction of infinite families of
bent functions by Mesnager \cite{M2014} and Xu et al. \cite{XCX2015} and construct  several new infinite families of
bent functions and bent idempotents from some known bent functions.

Let $n=2m$, $g(x)$ be a Boolean bent function over
$\fq$, $\tau$ be a positive  integer
such that $1\leq \tau \leq m$,   $X_1,X_2,\cdots,X_\tau$ be
$\tau$ variables, and
$F(X_1,X_2,\cdots,X_\tau)
\in \mathbb{F}_2[X_1,X_2,\cdots,X_\tau]$.
We will study Boolean functions
over $\fq$ of the form
\begin{equation*}
f(x)=g(x)+F(\Tr_1^n(u_1x),\Tr_1^n(u_2x)
,\cdots, \Tr_1^n(u_\tau x)),
\end{equation*}
where $u_i\in \fqm$.

From Lemma \ref{lem2.2}, as a function from $\mathbb{F}_{2}^{\tau}$
to the complex field $\mathbb{C}$,
$(-1)^{F(X_1,X_2,\cdots,X_\tau)}$ has the unique Fourier expansion, i.e., there exists
a unique set of
$c_{\mathbf{w}}\in \mathbb{C}$  such that
\begin{equation}\label{fourier}
(-1)^{F(X_1,X_2,\cdots,X_\tau)}
=\sum_{\mathbf{w}\in \mathbb{F}_2^{\tau}}
c_{\mathbf{w}}(-1)^{w_1X_1+w_2X_2+
\cdots +w_\tau X_\tau},
\end{equation}
where $\mathbf{w}=(w_1,w_2,\cdots,w_\tau)
\in \ftau$. Equation (\ref{fourier})
holds for any $X_1,X_2,\cdots,X_\tau\in
\mathbb{F}_2$. In particular, take
$X_1=\Tr_1^n(u_1 x)$, $X_2=\Tr_1^n(u_2 x)$,
$\cdots$, $X_\tau=\Tr_1^n(u_\tau x)$. Then for any $x\in \fq$, we have
\begin{equation}\label{fourier1}
(-1)^{F(\Tr_1^n(u_1 x),\Tr_1^n(u_2 x),\cdots,\Tr_1^n(u_\tau x))}
=\sum_{\mathbf{w}\in \mathbb{F}_2^{\tau}}
c_{\mathbf{w}}(-1)^{\Tr_1^n((\sum_{i=1}^{
\tau}w_iu_i)x)}.
\end{equation}
Multiplying both sides of
Equation (\ref{fourier1}) by
$(-1)^{g(x)+\Tr_1^n(\beta x)}$, we have
$$
(-1)^{g(x)+F(\Tr_1^n(u_1 x),\Tr_1^n(u_2 x),\cdots,\Tr_1^n(u_\tau x))+\Tr_1^n(\beta x)}
=\sum_{\mathbf{w}\in \mathbb{F}_2^{\tau}}
c_{\mathbf{w}}(-1)^{g(x)+\Tr_1^n((\beta+
\sum_{i=1}^{
\tau}w_iu_i)x)}.
$$
Further, we have
$$
\sum_{x\in \fq}(-1)^{g(x)+F(\Tr_1^n(u_1 x),\Tr_1^n(u_2 x),\cdots,\Tr_1^n(u_\tau x))+\Tr_1^n(\beta x)}
=\sum_{\mathbf{w}\in \mathbb{F}_2^{\tau}}
c_{\mathbf{w}}\mathcal{W}_g(\beta+\sum_{i=1}^{
\tau}w_iu_i),
$$
i.e.,
$$
\mathcal{W}_f(\beta)
=\sum_{\mathbf{w}\in \mathbb{F}_2^{\tau}}
c_{\mathbf{w}}\mathcal{W}_g(\beta+\sum_{i=1}^{
\tau}w_iu_i).
$$
Let $\widetilde{g}(x)$ be the dual of $g(x)$. Then
\begin{equation}\label{walshf}
\mathcal{W}_f(\beta)
=2^m\sum_{\mathbf{w}\in \mathbb{F}_2^{\tau}}
c_{\mathbf{w}}(-1)^{\widetilde{g}
(\beta+\sum_{i=1}^{
\tau}w_iu_i)}.
\end{equation}
For different bent functions
$g(x)$, we compute $\mathcal{W}_f(\beta)$ and construct new infinite families of bent functions.

\subsection{New bent functions from Kasami functions}\label{subA}
Let $n=2m~(m\geq 2)$. The Kasami function
\cite{M2014} is a bent function of the form $g(x)=\Tr_1^m(\lambda x^{2^m+1})$. And its dual is  $\widetilde{g}(x)
=\Tr_1^m(\lambda^{-1} x^{2^m+1})+1$.
Then we will characterize the bentness of functions of the form
\begin{equation}\label{kasami}
f(x)=\Tr_1^m(\lambda x^{2^m+1})+F(\Tr_1^n(u_1x),\Tr_1^n(u_2x)
,\cdots, \Tr_1^n(u_\tau x)),
\end{equation}
where $u_i\in \fqm$ and
$F(X_1,X_2
,\cdots, X_\tau)$ is a reduced polynomial.
In the following theorem, we present a
construction of a new family of bent functions from the Kasami functions and
compute their duals.

\begin{theorem}\label{at1}
Let $n=2m$, $\tau$ be a positive integer
such that $1\leq \tau \leq m$, $\lambda\in \mathbb{F}_{2^m}^{\times}$,  and $u_1,
u_2,\cdots,u_\tau\in \mathbb{F}_{2^n}$. If
$\Tr_1^n(\lambda^{-1} u_i^{2^m}
u_j)=0$ for any $1\leq i<j\leq \tau$, the function
defined in Equation (\ref{kasami}) is bent and its dual is
$$
\widetilde{f}(x)
=\Tr_1^m(\lambda x^{2^m+1})+F(X_1,X_2,\cdots,X_\tau)+1,
$$
where $X_i=\Tr_1^m(\lambda^{-1}
(x^{2^m}u_i+x u_i^{2^m}+u_i^{2^m+1}))$.
\end{theorem}
\begin{proof}
From Equation (\ref{walshf}),
\begin{align*}
\mathcal{W}_f(\beta)
=&2^m\sum_{\mathbf{w}\in \mathbb{F}_2^{\tau}}
c_{\mathbf{w}}(-1)^{\Tr_1^m(\lambda^{-1}
(\beta+\sum_{i=1}^{
\tau}w_iu_i)^{2^m+1})+1}\nonumber\\
=&-2^m\sum_{\mathbf{w}\in \mathbb{F}_2^{\tau}}
c_{\mathbf{w}}(-1)^{\Tr_1^m(\lambda^{-1}
(\beta+\sum_{i=1}^{
\tau}w_iu_i)^{2^m+1})}.
\end{align*}
From properties of the trace function,
\begin{align*}
\Tr_1^m(\lambda^{-1}
(\beta+\sum_{i=1}^{
\tau}w_iu_i)^{2^m+1})=&
\Tr_1^m(\lambda^{-1}\beta^{2^m+1})
+\sum_{i=1}^{\tau}w_i\Tr_1^m(\lambda^{-1}
(\beta^{2^m}u_i+\beta u_i^{2^m}+u_i^{2^m+1}))\\
&+\sum_{1\leq i<j\leq \tau}w_iw_j
\Tr_1^m(\lambda^{-1}(u_i^{2^m}u_j+
u_iu_j^{2^m})).
\end{align*}
From $\Tr_1^n(\lambda^{-1}u_i^{2^m}u_j)
=\Tr_1^m(\lambda (u_i^{2^m}u_j+
u_iu_j^{2^m}))=0$,
\begin{equation*}
\Tr_1^m(\lambda^{-1}
(\beta+\sum_{i=1}^{
\tau}w_iu_i)^{2^m+1})=
\Tr_1^m(\lambda^{-1}\beta^{2^m+1})
+\sum_{i=1}^{\tau}w_iX_i,
\end{equation*}
where $X_i=\Tr_1^m(\lambda^{-1}
(\beta^{2^m}u_i+\beta u_i^{2^m}+u_i^{2^m+1}))$.
From Equation (\ref{walshf}), we have
$$
\mathcal{W}_f(\beta)
=-2^m
(-1)^{\Tr_1^m(\lambda \beta^{2^m+1})}\sum_{\mathbf{w}\in \mathbb{F}_2^{\tau}}
c_{\mathbf{w}}(-1)^{\sum_{i=1}^{\tau}w_iX_i}.
$$
From Equation (\ref{fourier}),
$$
\mathcal{W}_f(\beta)
=-2^m
(-1)^{\Tr_1^m(\lambda \beta^{2^m+1})}
(-1)^{F(X_1,X_2,\cdots,X_\tau)}=
2^m
(-1)^{\Tr_1^m(\lambda \beta^{2^m+1})+F(X_1,X_2,\cdots,X_\tau)+1}
.
$$
Hence, $f(x)$ is bent and its dual is
$$
\widetilde{f}(x)
=\Tr_1^m(\lambda x^{2^m+1})+F(X_1,X_2,\cdots,X_\tau)+1,
$$
where $X_i=\Tr_1^m(\lambda^{-1}
(x^{2^m}u_i+x u_i^{2^m}+u_i^{2^m+1}))$.
\end{proof}
\begin{corollary}\label{ac2}
Let $n=2m$, $\tau$ be a positive integer such that  $1\leq \tau\leq m$, $\lambda, u_1,
u_2,\cdots,u_\tau\in \mathbb{F}_{2^m}^{\times}$, and $F(X_1,X_2,\cdots,X_\tau)$ be a reduced
polynomial in $\mathbb{F}_2[X_1,X_2,\cdots,X_\tau]$ of  algebraic degree $d$,
where $u_1,u_2,\cdots,u_\tau$ are linearly independent over $\mathbb{F}_2$. Then the function $f(x)$ defined in  Equation  (\ref{kasami}) is bent and  its dual is
$$
\widetilde{f}(x)=\Tr_1^m(\lambda^{-1} x^{2^m+1})+F(\Tr_1^n(\lambda^{-1}u_1x)
+\Tr_1^m(\lambda^{-1}u_1^2),\Tr_1^n(\lambda^{-1}u_2x)
+\Tr_1^m(\lambda^{-1}u_2^2)
,\cdots, \Tr_1^n(\lambda^{-1}u_\tau x)
+\Tr_1^m(\lambda^{-1}u_\tau^2))+1.
$$
If $d\geq 2$,  then $f(x)$ is a bent function of algebraic degree $d$.
\end{corollary}
\begin{proof}
Note that $u_i\in \mathbb{F}_{2^m}$ and
$\Tr_1^n(\lambda u_i^{2^m}u_j)
=\Tr_1^n(\lambda u_iu_j)=0$.
From Theorem \ref{at1}, $f(x)$ is a bent function.
Note that $\Tr_1^n(\lambda^{-1} u_i^{2^m}x)
=\Tr_1^n(\lambda^{-1} u_i x)$ and
$\Tr_1^n(\lambda^{-1} u_i^{2^m+1})
=\Tr_1^n(\lambda^{-1} u_i^2)$.
From Theorem \ref{at1}, we have
$$
\widetilde{f}(x)=\Tr_1^m(\lambda^{-1} x^{2^m+1})+F(\Tr_1^n(\lambda^{-1}u_1x)
+\Tr_1^m(\lambda^{-1}u_1^2),\Tr_1^n(\lambda^{-1}u_2x)
+\Tr_1^m(\lambda^{-1}u_2^2)
,\cdots, \Tr_1^n(\lambda^{-1}u_\tau x)
+\Tr_1^m(\lambda^{-1}u_\tau^2))+1.
$$
The algebraic degree of  $F(X_1,X_2,\cdots,X_\tau)$
is $d$ and the algebraic degree of
$\Tr_1^m(\lambda^{-1} x^{2^m+1})$ is $2$. If $d\geq 2$, from Lemma \ref{degree}, the algebraic degree of $f(x)$ is $d$.

Hence, this corollary follows.
\end{proof}

\begin{theorem}\label{at3}
Let $n=2m$, $u$ be a normal element of $\mathbb{F}_{2^m}$,  and $F(X_1,X_2,\cdots,X_m)$ be a reduced rotation  symmetric
polynomial in $\mathbb{F}_2[X_1,X_2,\cdots,X_m]$ of algebraic degree $d$. Then
the function
$$
f(x)=\Tr_1^m( x^{2^m+1})+F(\Tr_1^n(ux),\Tr_1^n(u^2x)
,\cdots, \Tr_1^n(u^{2^{m-1}} x))
$$
is a bent idempotent. And its dual
$$
\widetilde{f}(x)=\Tr_1^m( x^{2^m+1})+F(\Tr_1^n(ux)+1,\Tr_1^n(u^2x)+1
,\cdots, \Tr_1^n(u^{2^{m-1}} x)+1)+1
$$
is also a bent idempotent.
If $d\geq 2$, $f(x)$ is a bent idempotent of  algebraic degree $d$.
\end{theorem}
\begin{proof}
Since $u$ is a normal element of
$\mathbb{F}_{2^m}$, then $u,u^2,\cdots,u^{2^{m-1}}$ is a basis of
$\mathbb{F}_{2^m}$ over $\mathbb{F}_2$ and
$\Tr_1^m(u)=1$.
From Corollary \ref{ac2},  $f(x)$ is a bent function  and
$\widetilde{f}(x)$ is also bent.
From the rotation symmetric property of $F(X_1,X_2,\cdots,X_\tau)$,
\begin{align*}
f(x^2)=& \Tr_1^m( x^{2(2^m+1)})+F(\Tr_1^n(ux^2)+1,\Tr_1^n(u^2x^2)
+1,\cdots, \Tr_1^n(u^{2^{m-1}} x^2)+1)\\
=& \Tr_1^m( x^{2^m+1})+F(\Tr_1^n(u^{2^{n-1}}x)+1,
\Tr_1^n(u^{2n}x^2)+1
,\cdots, \Tr_1^n(u^{2^{m+n-2}} x)+1)\\
=& \Tr_1^m( x^{2^m+1})+F(\Tr_1^n(u^{2^{m-1}}x)+1,
\Tr_1^n(ux^2)+1
,\cdots, \Tr_1^n(u^{2^{m-2}} x)+1)\\
=& \Tr_1^m( x^{2^m+1})+F(\Tr_1^n(ux)+1,
\Tr_1^n(u^2x)+1
,\cdots, \Tr_1^n(u^{2^{m-1}} x)+1)\\
=& f(x).
\end{align*}
Similarly, $\widetilde{f}(x^2)=
\widetilde{f}(x)$.
$f(x)$ and $\widetilde{f}(x)$ are both idempotents. If $d\geq 2$, from Lemma \ref{degree}, then
$f(x)$ is a bent idempotent of algebraic degree $d$.

Hence, this theorem follows.
\end{proof}

\begin{corollary}\label{ac4}
Let $n=2m$,   $u$ be a normal element of $\mathbb{F}_{2^m}$, and $d$ be a positive integer such that $2\leq d\leq m$. Then
the function
$$
f(x)=\Tr_1^m( x^{2^m+1})+
\sum_{0\leq i_1<i_2<\cdots<i_d\leq m-1}
\prod_{j=1}^{d}(\Tr_1^n(u^{2^{i_j}}x))
$$
is a bent idempotent of algebraic degree $d$.
\end{corollary}
\begin{proof}
From Theorem \ref{at3}, we just take
$F(X_1,X_2,\cdots,X_m)$ as the $d$-th
elementary
symmetric polynomial  and obtain this corollary.
\end{proof}
\begin{corollary}\label{ac5}
Let $n=2m$, $\lambda\in \mathbb{F}_{2^m}^{\times}$,
$u_1,
u_2,\cdots,u_m$ be a basis of
$\mathbb{F}_{2^m}$ over $\mathbb{F}_2$, $F(X_1,X_2,\cdots,X_m)$ be a reduced
polynomial in  $\mathbb{F}_2[X_1,X_2,\cdots,X_m]$ of algebaric degree $d$. Then the function
 $$
f(x)=\Tr_1^m(\lambda x^{2^m+1})+F(\Tr_1^n(u_1x),\Tr_1^n(u_2x)
,\cdots, \Tr_1^n(u_m x))
$$
is bent and  its dual is
$$
\widetilde{f}(x)=\Tr_1^m(\lambda^{-1} x^{2^m+1})+F(\Tr_1^n(\lambda^{-1}u_1x)
+\Tr_1^m(\lambda^{-1}u_1^2),\Tr_1^n(\lambda^{-1}u_2x)
+\Tr_1^m(\lambda^{-1}u_2^2)
,\cdots, \Tr_1^n(\lambda^{-1}u_m x)
+\Tr_1^m(\lambda^{-1}u_m^2))+1.
$$
If $d\geq 2$, then $f(x)$ is a bent function of algebraic degree $d$.
\end{corollary}
\begin{proof}
From Corollary \ref{ac2} and
$\tau=m$, this corollary
follows.
\end{proof}
\begin{corollary}\label{ac6}
Let $n=2m$, $\tau$ be a positive integer such that $2\leq \tau\leq m$, and $\lambda, u_1,
u_2,\cdots,u_\tau\in \mathbb{F}_{2^m}^{\times}$,
where $u_i$ are linearly independent over $\mathbb{F}_2$.  Then the function
$$
{f}(x)=\Tr_1^m(\lambda  x^{2^m+1})+\prod_{i=1}^{\tau}
\Tr_1^n(u_ix)
$$
is bent of
algebraic degree $\tau$ and its dual is
$$
\widetilde{f}(x)=\Tr_1^m(\lambda^{-1} x^{2^m+1})+\prod_{i=1}^{\tau}
(\Tr_1^n(\lambda^{-1}u_ix)+
\Tr_1^m(\lambda^{-1}u_i^2))+1.
$$
\end{corollary}
\begin{proof}
From Corollary \ref{ac2}, this corollary follows.
\end{proof}

Let $T_0=\{x\in \mathbb{F}_{2^m}:
\Tr_1^m(x)=0\}$. Then the dimension of $T_0$ over $\mathbb{F}_2$ is $m-1$.
The following corollary gives a class of
self-dual bent functions.
\begin{corollary}\label{ac7}
Let $n=2m$,
$u_1,
u_2,\cdots,u_{m-1}$ be  a basis of
$T_0$ over $\mathbb{F}_2$, and
$F(X_1,X_2,\cdots,X_{m-1})$ be a reduced
polynomial in  $\mathbb{F}_2[X_1,X_2,\cdots,X_{m-1}]$ of algebraic degree $d$. Then the function
$$
f(x)=\Tr_1^m( x^{2^m+1})+F(\Tr_1^n(u_1x),\Tr_1^n(u_2x)
,\cdots, \Tr_1^n(u_{m-1} x))
$$
is anti-self-dual bent. And its dual is
$$
\widetilde{f}(x)=\Tr_1^m( x^{2^m+1})+F(\Tr_1^n(u_1x)
,\Tr_1^n(u_2x)
,\cdots, \Tr_1^n(u_{m-1} x)
)+1.
$$
If $d\geq 2$, then $f(x)$ is an  anti-self-dual bent function of algebraic degree $d$.
\end{corollary}
\begin{proof}
Since $u_i\in T_0$, then $\Tr_1^m(u_i^2)
=\Tr_1^m(u_i)=0$.
From Corollary \ref{ac2}, this corollary follows.
\end{proof}
\begin{remark}
Let $u_1,
u_2,\cdots,u_{m-1}$ be  a basis of
$T_0$ over $\mathbb{F}_2$. Take three
reduced polynomials $F_1(X_1,\cdots,
X_{m-1})$, $F_2(X_1,\cdots,
X_{m-1})$, and $F_3(X_1,\cdots,
X_{m-1})$. From Corollary \ref{ac7},
the three functions
\begin{align*}
& f_1(x)=\Tr_1^m( x^{2^m+1})+F_1(\Tr_1^n(u_1x),\Tr_1^n(u_2x)
,\cdots, \Tr_1^n(u_{m-1} x)),\\
& f_2(x)=\Tr_1^m( x^{2^m+1})+F_2(\Tr_1^n(u_1x),\Tr_1^n(u_2x)
,\cdots, \Tr_1^n(u_{m-1} x)),
\\
& f_3(x)=\Tr_1^m( x^{2^m+1})+F_3(\Tr_1^n(u_1x),\Tr_1^n(u_2x)
,\cdots, \Tr_1^n(u_{m-1} x))
\end{align*}
are anti-self-dual  bent functions.
From Corollary \ref{ac7}, $f_1
+f_2+f_3$ is also anti-self-dual bent. This solves the open problem proposed by
Mesnager  \cite{M2014}.
\end{remark}

\subsection{New bent functions and bent idempotents from  quadratic bent idempotents}\label{quadratic}

Let $g(x)$ be a quadratic idempotent bent function of the form
\begin{equation}\label{qibf}
g(x)=\sum_{i=0}^{m-1}c_i\Tr_1^n(x^{2^i+1})
+c_m\Tr_1^m(x^{2^m+1})+\varepsilon,
\end{equation}
where $\varepsilon\in \mathbb{F}_2$ and
$c_i\in \mathbb{F}_2~(i=0,1,\cdots,m)$.
The dual of a quadratic bent function is also a quadratic bent function. Let
$\widetilde{g}(x)$ be the dual of $g(x)$.
The following lemma shows that if
$g(x)$ is idempotent, then
$\widetilde{g}(x)$ is also idempotent.
\begin{lemma}\label{idempotent}
Let $g(x)$ be a bent idempotent. Then
$\widetilde{g}(x)$ is also a bent idempotent.
\end{lemma}
\begin{proof}
The Walsh transform of $g$ is
\begin{align*}
\mathcal{W}_g(\beta)=& \sum_{x\in \mathbb{F}_q}(-1)^{g(x)+\Tr_1^n(\beta x)}\\
=& \sum_{x\in \mathbb{F}_q}(-1)^{g(x^2)+\Tr_1^n(\beta x^2)}\\
=&\sum_{x\in \mathbb{F}_q}(-1)^{g(x)+\Tr_1^n(
\beta^{2^{n-1}} x)}.
\end{align*}
Then $\mathcal{W}_g(\beta)
=\mathcal{W}_g(\beta^{2^{n-1}})$ for any
$\beta$. Hence,
$\mathcal{W}_g(\beta^2)
=\mathcal{W}_g((\beta^2)^{2^{n-1}})
=\mathcal{W}_g(\beta)$. Since
$g(x)$ is bent, then
$$
2^m(-1)^{\widetilde{g}(\beta^2)}
=2^m(-1)^{\widetilde{g}(\beta)},
$$
i.e., $\widetilde{g}(\beta^2)
=\widetilde{g}(\beta)$.
Hence, $\widetilde{g}(x)$ is also
idempotent bent.
\end{proof}

From Lemma \ref{idempotent}, we have
\begin{equation*}
\widetilde{g}(x)=\sum_{i=0}^{m-1}
\widetilde{c}_i\Tr_1^n(x^{2^i+1})
+\widetilde{c}_m\Tr_1^m(x^{2^m+1})
+\widetilde{\varepsilon}.
\end{equation*}

In this subsection, we consider  functions of the form
\begin{equation}\label{q-bent}
f(x)=g(x)+F(\Tr_1^n(u_1x),\Tr_1^n(u_2x)
,\cdots, \Tr_1^n(u_\tau x)),
\end{equation}
where $g(x)$ is defined in
Equation (\ref{qibf}),
$\tau$ is a positive integer
such that $1\leq \tau \leq m$,
$u_1,u_2,\cdots,u_\tau\in \mathbb{F}_{2^m}^{\times}$, and
$F(X_1,X_2,\cdots,X_\tau)$ is a reduced
polynomial in $\mathbb{F}_2[X_1,X_2,\cdots,X_\tau]$.
\begin{theorem}\label{thm-i}
Let $n=2m$,  $\tau$ be a positive integer
such that $1\leq \tau \leq m$, $g(x)$ be quadratic idempotent bent
of the form (\ref{qibf}),
$u_1,u_2,\cdots,u_\tau\in \mathbb{F}_{2^m}^{\times}$,
and $F(X_1,X_2,\cdots,X_\tau)$ be a reduced
polynomial in  $\mathbb{F}_2[X_1,X_2,\cdots,X_\tau]$.
Then the function $f(x)$ defined in Equation (\ref{q-bent}) is bent.
\end{theorem}
\begin{proof}
From Equation (\ref{walshf}),
\begin{equation*}
\mathcal{W}_f(\beta)
=2^m\sum_{\mathbf{w}\in \mathbb{F}_2^{\tau}}
c_{\mathbf{w}}(-1)^{\widetilde{g}
(\beta+\sum_{i=1}^{
\tau}w_iu_i)}.
\end{equation*}
When $u\in \mathbb{F}_{2^m}$,
we have $\Tr_1^n((\beta
+u)^{2^i+1})=
\Tr_1^n(\beta^{2^i+1})
+\Tr_1^n((\beta^{2^i}
+\beta^{2^{n-i}})u)$ and
$\Tr_1^m((\beta
+u)^{2^m+1})=
\Tr_1^m(\beta^{2^m+1})
+\Tr_1^m((\beta^{2^m}
+\beta+1)u)$.
Let $u=\sum_{i=1}^{\tau}w_iu_i$. Then
$$
\widetilde{g}(\beta+u)=
\widetilde{g}(\beta)+
\sum_{i=1}^{\tau}w_iX_i.
$$
where $X_i=\Tr_1^n(u_i\sum_{j=0}^{m-1}
\widetilde{c}_j(\beta^{2^j}+
\beta^{2^{n-j}}))+\Tr_1^m(u_m(\beta^{2^m}
+\beta+1))$. Then
$$
\mathcal{W}_f(\beta)
=2^m(-1)^{\widetilde{g}(\beta)}
\sum_{\mathbf{w}\in \mathbb{F}_2^{\tau}}
c_{\mathbf{w}}(-1)^{\sum_{i=1}^{
\tau}w_iX_i}.
$$
From Equation (\ref{fourier}),
$$
\mathcal{W}_f(\beta)
=2^m(-1)^{\widetilde{g}(\beta)
+F(X_1,X_2,\cdots,X_\tau)}.
$$
Hence, $f(x)$ is bent.
\end{proof}

\begin{theorem}\label{rs-bent}
Let $n=2m$, $g(x)$ be quadratic idempotent bent of the form (\ref{qibf}), $u$ be a normal element of
$\mathbb{F}_{2^m}$, and
$F(X_1,X_2,\cdots,X_m)$ be a reduced
polynomial of algebraic degree $d$ in  $\mathbb{F}_2[X_1,X_2,\cdots,X_m]$.
Then the function
$$
f(x)=g(x)+F(\Tr_1^n(ux),\Tr_1^n(u^2x)
,\cdots, \Tr_1^n(u^{2^{m-1}} x)),
$$ is always bent. Further,
if $F(X_1,X_2,\cdots,X_m)$ is
rotation symmetric, then
$f(x)$ is a bent idempotent. And when  $d\geq 2$, $f(x)$ is a bent idempotent of   algebraic degree $d$.
\end{theorem}
\begin{proof}
From Theorem \ref{thm-i},
$f(x)$ is bent.
 Since
$u$ is a normal element of
$\mathbb{F}_{2^m}$ and  $F(X_1,X_2,\cdots,X_m)$ is
rotation symmetric, we can verify that
$$
f(x^2)=f(x).
$$
Hence, $f$ is idempotent.
When $d\geq 2$, from Lemma \ref{degree},
the algebraic degree of $f(x)$ is $d$.
\end{proof}
\begin{remark}
The quadratic idempotent bent function $g(x)$ is
required in the construction.
Ma et al. \cite{MLZ2005} proved that
$g(x)$ is bent if and only if
$gcd(\sum_{i=1}^{m-1}c_i(X^i+X^{n-i})
+c_mX^m,X^n+1)=1$. This condition
is equivalent to that
the linearized polynomial $L(x)
=\sum_{i=1}^{m-1}c_i
(x^{2^i}+x^{2^{n-i}})+c_mx^{2^m}$ is a permutation polynomial (a necessary condition is that $L(1)\neq 0$, i.e.,
$c_m=1$).

If $n$ is a power of 2, then
according to \cite{SSS2009}[Proposition 3.1],
$g(x)$ is bent if and only if $c_m=1$.
For any $c_i\in \mathbb{F}_2$ and
rotation symmetric polynomial
$F(X_1,X_2,\cdots,X_n)$, the function in
Theorem \ref{rs-bent} of the form
$$
f(x)=\sum_{i=1}^{m-1}c_i
\Tr_1^n(x^{2^i+1})+\Tr_1^m(x^{2^m+1})+
F(\Tr_1^n(u x),\Tr_1^n(u^2x)
,\cdots, \Tr_1^n(u^{2^{m-1}} x))
$$
is always idempotent bent.

Note that more results, valid
for more general values of n, can be found in \cite{YG2006}.

\end{remark}
\subsection{New bent functions from Gold-like monomial functions}\label{subB}

Mesnager \cite{M2014} showed that
the monomial function
$g(x)=\Tr_1^{4k}(\lambda x^{2^k+1})$ over
$\mathbb{F}_{2^{4k}}$ is self-dual bent, i.e., for any $\beta\in \mathbb{F}_{2^{4k}}$,
the dual $\widetilde{g}(\beta)=
\Tr_1^{4k}(\lambda x^{2^k+1})$, where
$\lambda+\lambda^{2^{3k}}=1$.

In the following theorem, we present a
construction of a new family of bent functions from the Gold-like monomial  function and
compute their duals.
\begin{theorem}\label{bt1}
Let $n=4k$, $k$   be positive integers, $\tau$ be a positive integer
such that $1\leq \tau \leq 2k$,  $\lambda\in \mathbb{F}_{2^{4k}}$
such that $\lambda+\lambda^{2^{3k}}=1$,
$u_1,
u_2,\cdots,u_\tau\in \mathbb{F}_{2^{4k}}^{
\times}$,  and
$F(X_1,X_2,\cdots,X_\tau)$ be a reduced
polynomial in   $\mathbb{F}_2[X_1,X_2,\cdots,X_\tau]$. If
$\Tr_1^{4k}(\lambda(u_i^{2^k}u_j+
u_iu_j^{2^k}))=0$ for any
$1\leq i<  j\leq \tau$,  then the function $f(x)$
$$
f(x)=\Tr_1^{4k}(\lambda x^{2^k+1})+F(\Tr_1^{4k}(u_1x),\Tr_1^{4k}(u_2x)
,\cdots, \Tr_1^{4k}(u_\tau x))
$$
is bent.
\end{theorem}
\begin{proof}
From Equation (\ref{walshf}),
$$
\mathcal{W}_{f}(\beta)
=2^{2k}\sum_{\mathbf{w}\in \mathbb{F}_2^{\tau}}
c_{\mathbf{w}}(-1)^{\Tr_1^{4k}(\lambda
(\beta+\sum_{i=1}^{
\tau}w_iu_i)^{2^k+1})}.
$$
Note that
$$
\Tr_1^{4k}(\lambda
(\beta+\sum_{i=1}^{
\tau}w_iu_i)^{2^k+1})
=\Tr_1^{4k}(\lambda\beta^{2^k+1})
+\sum_{i=1}^{\tau}w_i\Tr_1^{4k}(\lambda
(\beta^{2^k}u_i+\beta u_i^{2^k}+u_i^{
2^k+1}))+
\sum_{1\leq i< j\leq \tau}
w_iw_j\Tr_1^{4k}(\lambda(u_i^{2k}u_j
+u_iu_j^{2^k})).
$$
From $\Tr_1^{4k}(\lambda(u_i^{2k}u_j
+u_iu_j^{2^k}))=0$,
$$
\mathcal{W}_f(\beta)
=2^{2k}(-1)^{\Tr_1^{4k}(\lambda
\beta^{2^k+1})}\sum_{\mathbf{w}
\in \mathbb{F}_2^{\tau}}c_{\mathbf{w}}
(-1)^{\sum_{i=1}^{\tau}w_iX_i},
$$
where $X_i=\Tr_1^{4k}(\lambda
(\beta^{2^k}u_i+\beta u_i^{2^k}+u_i^{
2^k+1}))$.
From Equation (\ref{fourier}),
$$
\mathcal{W}_f(\beta)
=2^{2k}(-1)^{\Tr_1^{4k}(\lambda
\beta^{2^k+1})+F(X_1,X_2,\cdots, X_\tau)}.
$$
Hence, $f(x)$ is bent and its dual is
$$
\widetilde{f}(x)
=\Tr_1^{4k}(\lambda
x^{2^k+1})+F(X_1,X_2,\cdots, X_\tau),
$$
where $X_i=\Tr_1^{4k}(\lambda
(x^{2^k}u_i+x u_i^{2^k}+u_i^{
2^k+1}))$.
\end{proof}

\subsection{New bent functions from Niho exponents}\label{subC}

Leander and Kholosha \cite{LK2006} introduced
a class of bent functions of the form
$$
g(x)=\Tr_1^m(x^{2^m+1})+
\Tr_1^n(\sum_{i=1}^{2^{k-1}-1}x^{(2^m-1)
i/2^k+1}),
$$
where $(k,m)=1$.
Its dual \cite{BCHKM2012} is
\begin{equation*}
\widetilde{g}(x)=
\Tr_1^m((\alpha A+x^{2^m}
+\alpha^{2^{n-k}})A^{1/2^{k-1}}),
\end{equation*}
where $\alpha+\alpha^{2^m}=1$ and
$A=1+x+x^{2^m}$.

In the following theorem, we present a
construction of a new family of bent functions from the Niho exponents and compute their duals.
\begin{theorem}\label{ct1}
Let $n=2m$, $k$ be a positive integer
satisfying $(k,m)=1$,
$\tau$ be a positive integer
such that $1\leq \tau \leq m$,
$u_1,
u_2,\cdots,u_\tau\in \mathbb{F}_{2^{m}}^{
\times}$, and
$F(X_1,X_2,\cdots,X_\tau)$ be a reduced
polynomial in   $\mathbb{F}_2[X_1,X_2,\cdots,X_\tau]$.  Then the function
$$
f(x)=\Tr_1^m(x^{2^m+1})+
\Tr_1^n(\sum_{i=1}^{2^{k-1}-1}x^{(2^m-1)
i/2^k+1})+F(\Tr_1^{n}(u_1x),\Tr_1^{n}(u_2x)
,\cdots, \Tr_1^{n}(u_\tau x))
$$
is bent.
\end{theorem}
\begin{proof}
Note that
$$
1+(\beta+\sum_{i=1}^{\tau}w_iu_i)
+(\beta+\sum_{i=1}^{\tau}w_iu_i)^{2^m}
=1+\beta+\beta^{2^m}.
$$
From the dual $\widetilde{g}(x)$,
$$
\widetilde{g}(
\beta+\sum_{i=1}^{\tau}w_iu_i)
=\widetilde{g}(\beta)
+\sum_{i=1}^{\tau}w_i\Tr_1^m(u_i
A^{1/2^{k-1}}),
$$
where $A=1+\beta+\beta^{2^m}$.
From Equation (\ref{walshf}), we have
\begin{align*}
\mathcal{W}_f(\beta)
=&2^m(-1)^{\widetilde{g}
(\beta)}\sum_{\mathbf{w}\in \mathbb{F}_2^{\tau}}
c_{\mathbf{w}}(-1)^{\sum_{i=1}^{
\tau}w_iX_i},
\end{align*}
where $X_i=
\Tr_1^m(u_iA^{1/2^{k-1}})$.
From Equation (\ref{fourier}),
$$
\mathcal{W}_f(\beta)
=2^{m}(-1)^{\widetilde{g}(\beta)
+F(X_1,X_2,\cdots, X_\tau)}.
$$
Hence, $f(x)$ is bent and its dual is
$$
\widetilde{f}(x)
=\widetilde{g}(x)
+F(\Tr_1^m(u_1A^{1/2^{k-1}}),
\Tr_1^m(u_2A^{1/2^{k-1}}),\cdots, \Tr_1^m(u_\tau A^{1/2^{k-1}})),
$$
where $A=1+x+x^{2^m}$.
\end{proof}
\begin{corollary}\label{cc1}
Let $n=2m$, $k$ be a positive integer
satisfying $(k,m)=1$, $u$ be a normal  element of $\mathbb{F}_{2^m}$, and
$F(X_1,X_2,\cdots,X_m)$ be a reduced
polynomial in   $\mathbb{F}_2[X_1,X_2,\cdots,X_m]$. Then the function
$$
f(x)=\Tr_1^m(x^{2^m+1})+
\Tr_1^n(\sum_{i=1}^{2^{k-1}-1}x^{(2^m-1)
i/2^k+1})+F(\Tr_1^{n}(ux),\Tr_1^{n}(u^2x)
,\cdots, \Tr_1^{n}(u^{2^{m-1}} x))
$$
is bent. Further, if $F(X_1,X_2,\cdots,X_m)$ is rotation symmetric, $f(x)$ is a bent idempotent.
\end{corollary}
\begin{proof}
From Theorem \ref{ct1}, $f(x)$ is bent.
From a similar proof of
Theorem \ref{at3}, we have
$f(x^2)=f(x)$. Hence,
$f(x)$ is a bent idempotent.
\end{proof}

\subsection{New bent functions from the
Maiorana-McFarland class of bent functions }\label{subD}

The Maiorana-McFarland class of bent functions are defined over $\mathbb{F}_{2^m}\times \mathbb{F}_{2^m}$ of the form
\begin{equation*}
g(x,y)=\Tr_1^m(x\pi(y))+h(y),
\end{equation*}
where $\pi: \mathbb{F}_{2^m}\longrightarrow \mathbb{F}_{2^m}$ is a permutation and
$h$ is a Boolean function over $\mathbb{F}_{2^m}$. Its Walsh transform is
$$
\mathcal{W}_{g}(\beta_1,\beta_2)
=2^m(-1)^{\Tr_1^m(\beta_2\pi^{-1}(\beta_1))
+h(\pi^{-1}(\beta_1))}.
$$
And its dual \cite{C2010} is
$$
\widetilde{g}(x,y)=
\Tr_1^m(y\pi^{-1}(x))+h(\pi^{-1}(x)).
$$

In the following theorem, we present a
construction of a new family of bent functions from the Maiorana-McFarland
class of bent functions and
compute their duals.
\begin{theorem}
Let $n=2m$, $\tau$ be a positive integer
such that $1\leq \tau \leq m$,  $u_i=(u_i^{(1)},u_i^{(2)})
\in \mathbb{F}_{2^m}\times
\mathbb{F}_{2^m}~~(1\leq i\leq \tau)$ be linearly independent over $\mathbb{F}_2$, $\pi$ be a linear permutation polynomial
over $\mathbb{F}_{2^m}$, and
$F(X_1,X_2,\cdots,X_\tau)$ be a reduced
polynomial in  $\mathbb{F}_2[X_1,X_2,\cdots,X_\tau]$. If
$\Tr_1^{m}(u_i^{(2)}\pi^{-1}(u_j^{(1)})+
u_j^{(2)}\pi^{-1}(u_i^{(1)}))=0$ for any  $1\leq i<j\leq \tau$
Then the function
$$
f(x,y)=\Tr_1^m(x\pi(y))+\Tr_1^m(by)+
F(\Tr_1^{m}(u_1^{(1)}x+
u_1^{(2)}y),\Tr_1^{m}(u_2^{(1)}x+
u_2^{(2)}y)
,\cdots, \Tr_1^{m}(u_\tau^{(1)}x+
u_\tau^{(2)}y))
$$
is bent.
\end{theorem}
\begin{proof}
Note that
$$
\widetilde{g}(\beta_1+
\sum_{i=1}^{\tau}w_iu_i^{(1)},
\beta_2+
\sum_{i=1}^{\tau}w_iu_i^{(2)})
=\widetilde{g}(\beta_1,\beta_2)
+\sum_{i=1}^{\tau}w_iX_i
+\sum_{1\leq i< j\leq \tau}
w_iw_j\Tr_1^m(u_i^{(2)}\tau^{-1}(u_j^{(1)})
+u_j^{(2)}\tau^{-1}(u_i^{(1)})),
$$
where $X_i=
\Tr_1^m((\beta_2+b)
\pi^{-1}(u_i^{(1)})+u_i^{(2)}
\tau^{-1}(\beta_1)
+u_i^{(2)}\tau^{-1}(u_i^{(1)}))$.
From $\Tr_1^{m}(u_i^{(2)}\tau^{-1}(u_j^{(1)})
+u_j^{(2)}\tau^{-1}(u_i^{(1)}))=0
$ and Equation (\ref{walshf}), we have
$$
\mathcal{W}_f(\beta)
=2^m(-1)^{\widetilde{g}
(\beta_1,\beta_2)}\sum_{\mathbf{w}\in \mathbb{F}_2^{\tau}}
c_{\mathbf{w}}(-1)^{\sum_{i=1}^{
\tau}w_iX_i}.
$$
From Equation (\ref{fourier}),
$$
\mathcal{W}_f(\beta)
=2^{m}(-1)^{\widetilde{g}(\beta_1,\beta_2)
+F(X_1,X_2,\cdots, X_\tau)}.
$$
Hence, $f(x)$ is bent and its dual is
$$
\widetilde{f}(x,y)
=\Tr_1^m(y\pi^{-1}(x))+
\Tr_1^m(b\pi^{-1}(x))
+F(X_1,X_2,\cdots,X_\tau),
$$
where
$X_i=
\Tr_1^m((y+b)
\pi^{-1}(u_i^{(1)})+u_i^{(2)}
\tau^{-1}(x)
+u_i^{(2)}\tau^{-1}(u_i^{(1)}))$.
\end{proof}
\begin{theorem}
Let $n=2m$, $s$ be a divisor of $m$ such that $m\over s$ is odd, $\tau$ be a positive integer
such that $1\leq \tau \leq m$,  $u_i=(u_i^{(1)},u_i^{(2)})
\in \mathbb{F}_{2^s}\times
\mathbb{F}_{2^s}~~(1\leq i\leq \tau)$ be linearly independent over $\mathbb{F}_2$, and
$F(X_1,X_2,\cdots,X_\tau)$ be a reduced
polynomial in  $\mathbb{F}_2[X_1,X_2,\cdots,X_\tau]$,
where for any $1\leq i,j\leq \tau$,
$u_i^{(1)}u_j^{(2)}+
u_j^{(1)}u_i^{(2)}=0$
and
$\Tr_1^{m}((u_i^{(1)})^2u_j^{(2)}+
u_i^{(2)}(u_j^{(1)})^2=0$.
Then the function
$$
f(x,y)=\Tr_1^m(xy^d)+
F(\Tr_1^{m}(u_1^{(1)}x+
u_1^{(2)}y),\Tr_1^{m}(u_2^{(1)}x+
u_2^{(2)}y)
,\cdots, \Tr_1^{m}(u_\tau^{(1)}x+
u_\tau^{(2)}y))
$$
is bent,
where $d(2^s+1)\equiv 1 \mod 2^m-1$.
\end{theorem}
\begin{proof}
Take $\pi(y)=y^d$. From $d(2^s+1)\equiv 1\mod 2^m-1$,
$$
\widetilde{g}(\beta_1,\beta_2)
=\Tr_1^m(\beta_2\beta_1^{2^{s}+1}).
$$
From $u_i^{(1)}, u_i^{(2)}
\in \mathbb{F}_{2^s}$,
\begin{align*}
\widetilde{g}(\beta_1+
\sum_{i=1}^{\tau}w_iu_i^{(1)},
\beta_2+
\sum_{i=1}^{\tau}w_iu_i^{(2)})
=&\Tr_1^m(\beta_2\beta_1^{2^{s}+1})
+\sum_{i=1}^{\tau}w_iX_i\\
&+\sum_{1\leq i<  j\leq \tau}
w_iw_j\Tr_1^m((\beta_1^{2^s}+\beta_1)(
u_i^{(2)}u_j^{(1)}
+u_j^{(2)}u_i^{(1)})
+
u_i^{(2)}(u_j^{(1)})^2
+u_j^{(2)}(u_i^{(1)})^2),
\end{align*}
where $X_i=
\Tr_1^m(\beta_1^{2^s+1}u_i^{(2)}+
\beta_2(\beta_1^{2^s}u_i^{(1)}+\beta_1 u_i^{(1)}+(u_i^{(1)})^2)
+u_i^{(2)}(\beta_1^{2^{s}}u_i^{(1)}
+\beta_1u_i^{(1)}+(u_i^{(1)})^2))$.
From $u_i^{(2)}u_j^{(1)}
+u_j^{(2)}u_i^{(1)}=0$,  $
\Tr_1^m(u_i^{(2)}(u_j^{(1)})^2
+u_j^{(2)}(u_i^{(1)})^2)=0$,
and Equation (\ref{walshf}),
$$
\mathcal{W}_f(\beta)
=2^m(-1)^{\widetilde{g}
(\beta_1,\beta_2)}\sum_{\mathbf{w}\in \mathbb{F}_2^{\tau}}
c_{\mathbf{w}}(-1)^{\sum_{i=1}^{
\tau}w_iX_i}.
$$
From Equation (\ref{fourier}),
$$
\mathcal{W}_f(\beta)
=2^{m}(-1)^{\widetilde{g}(\beta_1,\beta_2)
+F(X_1,X_2,\cdots, X_\tau)}.
$$
Hence, $f(x)$ is bent.
\end{proof}

\section{Conclusion}
In this paper, we generalize previous work on the constructions of  infinite families of bent functions \cite{M2014,XCX2015} and
construct several infinite families of bent functions from known bent functions (the Kasami functions, quadratic idempotent bent functions, bent functions via Niho exponents, the Gold-like monomial functions, and the Maiorana-McFarland class of bent functions). These bent functions
contain some previous bent functions
by Mesnager \cite{M2014} and Xu et al.
\cite{XCX2015}. Further,
several   infinite families of
bent idempotents of any possible  algebraic degree
are obtained from any quadratic bent idempotent, which solves the open problem
proposed by Carlet \cite{C2014}.
And  an infinite family of anti-self-dual
bent functions are constructed. Take
three anti-self-dual bent functions
$f_1$, $f_2$, and $f_3$ in such a family, then $f_1+f_2+f_3$ belongs to this family, i.e., $f_1+f_2+f_3$ is anti-self-dual bent.
This solves the open problem proposed by
Mesnager \cite{M2014}.

\section*{Acknowledgment}
This work was supported by
the National Natural Science Foundation of China
(Grant No. 11401480, No.10990011 \& No. 61272499).
Yanfeng Qi also acknowledges support from
KSY075614050 of Hangzhou Dianzi University.


\ifCLASSOPTIONcaptionsoff
  \newpage
\fi


\begin{thebibliography}{99}
\bibitem{BCHKM2012} L. Budaghyan, C. Carlet, T. Helleseth, A. Kholosha, and S. Mesnager,
``Further results on Niho bent functions," IEEE Trans. Inf. Theory,
vol. 58, no. 11, pp. 6979-6985, Nov. 2012.
\bibitem{CK1986} R. Calderbank and W. M. Kantor, ``The geometry of two-weight codes,"
Bull. London Math. Soc., vol. 18, no. 2, pp. 97-122, 1986.
\bibitem{CC2003} A. Canteaut and P. Charpin, ``Decomposing bent functions," IEEE Trans.
Inf. Theory, vol. 49, no. 8, pp. 2004-2019, Aug. 2003.
\bibitem{CCK2008} A. Canteaut, P. Charpin, and G. Kyureghyan, ``A new class of monomial
bent functions," Finite Fields Their Appl., vol. 14, no. 1, pp. 221-241,
2008.
\bibitem{C1994} C. Carlet, ``Two new classes of bent functions," in EUROCRYPT
(Lecture Notes in Computer Science), vol. 765. New York, NY, USA:
Springer-Verlag, 1994, pp. 77-101.
\bibitem{C1996} C. Carlet, ``A construction of bent function," in Proc. 3rd Int. Conf.
Finite Fields and Appl., 1996, pp. 47-58.
\bibitem{C2006} C. Carlet, ``On bent and highly nonlinear balanced/resilient functions
and their algebraic immunities," in AAECC (Lecture Notes in Computer
Science), vol. 3857, M. P. C. Fossorier, H. Imai, S. Lin, and A. Poli, Eds.
New York, NY, USA: Springer-Verlag, 2006, pp. 1-28.
\bibitem{C2010} C. Carlet, ``Boolean functions for cryptography and error correcting
codes," in Boolean Models and Methods in Mathematics, Computer
Science, and Engineering, Y. Crama and P. L. Hammer, Eds. Cambridge,
U.K.: Cambridge Univ. Press, 2010, pp. 257-397.
\bibitem{C2014} C Carlet. ``Open Open Problems on Binary Bent Functions". Open Problems in Mathematics and Computational Science 2014, pp 203-241
\bibitem{CDPS2010} C. Carlet, L. E. Danielsen, M. G. Parker, and P. Sol¨¦, ``Self-dual bent
functions," Int. J. Inform. Coding Theory, vol. 1, no. 4, pp. 384-399,
2010.
\bibitem{CM2011} C. Carlet and S. Mesnager, ``On Dillon¡¯s class H of bent functions,
Niho bent functions and O-polynomials," J. Combinat. Theory, Ser. A,
vol. 118, no. 8, pp. 2392-2410, 2011.
\bibitem{CHKM2011} C. Carlet, T. Helleseth, A. Kholosha, and S. Mesnager, ``On the dual of
bent functions with 2r Niho exponents," in Proc. IEEE ISIT, Aug. 2011,
pp. 703-707.
\bibitem{CG2008} P. Charpin and G. Gong, ``Hyperbent functions, Kloosterman sums and
Dickson polynomials," in Proc. ISIT, Jul. 2008, pp. 1758-1762.
\bibitem{CK2008} P. Charpin and G. Kyureghyan, ``Cubic monomial bent functions:
A subclass of M," SIAM J. Discrete Math., vol. 22, no. 2, pp. 650-665,
2008.
\bibitem{CHLL1997} G. Cohen, I. Honkala, S. Litsyn, and A. Lobstein, Covering Codes.
Amsterdam, The Netherlands: North Holland, 1997.
\bibitem{D1974} J. Dillon, ``Elementary Hadamard difference sets," Ph.D. dissertation,
Netw. Commun. Lab., Univ. Maryland, College Park, MD, USA, 1974.


\bibitem{DD2004} J. F. Dillon and H. Dobbertin, ``New cyclic difference sets with Singer
parameters," Finite Fields Their Appl., vol. 10, no. 3, pp. 342-389, 2004.
\bibitem{Ding2015} C. Ding, ``Linear codes from some 2-designs," IEEE Trans. Inform. Theory, vol. 61, no. 6, pp. 3265-3275,  June 2015.
\bibitem{DLCCFG2006} H. Dobbertin, G. Leander, A. Canteaut, C. Carlet, P. Felke, and
P. Gaborit, ``Construction of bent functions via Niho power functions,"
J. Combinat. Theory, Ser. A, vol. 113, no. 5, pp. 779-798, 2006.

\bibitem{F1999} C. Fontaine, ``On some cosets of the first-order Reed-Muller code with high minimum weight," IEEE Trans. Inform. Theory 45, 1237-1243 (1999)
\bibitem{FF1998} E. Filiol, C. Fontaine, ``Highly nonlinear balanced Boolean functions with a good correlationimmunity," in Proceedings of EUROCRYPT¡¯98. Lecture Notes in Computer Science, vol. 1403 (1998), pp. 475-488

\bibitem{G1968} R. Gold, ``Maximal recursive sequences with 3-valued recursive crosscorrelation functions (Corresp.),"  IEEE Trans. Inf. Theory, vol. 14, no. 1,
pp. 154-156, Jan. 1968.
\bibitem{L2006} G. Leander, ``Monomial bent functions," IEEE Trans. Inf. Theory,
vol. 52, no. 2, pp. 738-743, Feb. 2006.
\bibitem{LK2006} G. Leander and A. Kholosha, ``Bent functions with 2r Niho exponents,"
IEEE Trans. Inf. Theory, vol. 52, no. 12, pp. 5529-5532, Dec. 2006.
\bibitem{LHTK2013} N. Li, T. Helleseth, X. Tang, and A. Kholosha, ``Several new classes of
bent functions from Dillon exponents," IEEE Trans. Inf. Theory, vol. 59,
no. 3, pp. 1818-1831, Mar. 2013.

\bibitem{MLZ2005} W. Ma, M. Lee, and F. Zhang, ``A new class of bent functions,¡° IEICE Trans. Fundamentals, vol. E88-A, no. 7, pp. 2039-2040, July 2005.

\bibitem{M1973} R. L. McFarland, ``A family of noncyclic difference sets," J. Combinat.
Theory, Ser. A, vol. 15, no. 1, pp. 1-10, 1973.
\bibitem{M} S. Mesnager, ``Bent functions from spreads," J. Amer. Math. Soc., to be published.
\bibitem{M2009} S. Mesnager, ``A new family of hyper-bent Boolean functions in polynomial form," in Cryptography and Coding (Lecture Notes in Computer
Science), vol. 5921, M. G. Parker, Ed. Berlin, Germany: Springer-Verlag,
2009, pp. 402-417.
\bibitem{M2010} S. Mesnager, ``Hyper-bent Boolean functions with multiple trace terms,"
in Arithmetic of Finite Fields (Lecture Notes in Computer Science),
vol. 6087, M. Hasan and T. Helleseth, Eds. Berlin, Germany:
Springer-Verlag, 2010, pp. 97-113.
\bibitem{M2011-1} S. Mesnager, ``Bent and hyper-bent functions in polynomial form and
their link with some exponential sums and Dickson polynomials,¡± IEEE
Trans. Inf. Theory, vol. 57, no. 9, pp. 5996-6009, Sep. 2011.
\bibitem{M2011-2} S. Mesnager, ``A new class of bent and hyper-bent Boolean functions
in polynomial forms," Des., Codes Cryptography, vol. 59, nos. 1-3,
pp. 265-279, 2011.
\bibitem{M2014} S. Mesnager, ``Several New Infinite Families of Bent Functions and Their Duals," IEEE Trans. Inf. Theory, vol. 60, no. 7, JULY 2014
\bibitem{MF2013} S. Mesnager and J. P. Flori, ``Hyper-bent functions via Dillon-like
exponents," IEEE Trans. Inf. Theory, vol. 59, no. 5, pp. 3215-3232,
May 2013.
\bibitem{N1993} K. Nyberg, ``On the construction of highly nonlinear permutations,"
in EUROCRYPT (Lecture Notes in Computer Science), vol. 658.
New York, NY, USA: Springer-Verlag, 1993, pp. 92-98.
\bibitem{OSW1982} J. D. Olsen, R. A. Scholtz, and L. R. Welch, ``Bent-function sequences,"
IEEE Trans. Inf. Theory, vol. 28, no. 6, pp. 858-864, Nov. 1982.
\bibitem{PQ1999} J. Pieprzyk, C. Qu, ``Fast Hashing and rotation symmetric functions," J. Univ. Comput. Sci. 5, 20-31 (1999)
\bibitem{PTF2010} A. Pott, Y. Tan, and T. Feng, ``Strongly regular graphs associated with
ternary bent functions," J. Combinat. Theory, Ser. A, vol. 117, no. 6,
pp. 668-682, 2010.
\bibitem{PTFL2011} A. Pott, Y. Tan, T. Feng, and S. Ling, ``Association schemes arising
from bent functions," Des., Codes Cryptography, vol. 59, nos. 1-3,
pp. 319-331, 2011.
\bibitem{R1976} O. Rothaus, ``On bent functions," J. Combinat. Theory, Ser. A, vol. 20, no. 3, pp. 300-305, 1976.
\bibitem{SSS2009} R. Singh, B. Sarma, A. Saikia, ``Public key cryptography using permutation p-polynomials over finite fields," IACR Cryptology ePrint Archive 2009: 208 (2009)

\bibitem{XCX2015} G. Xu, X. Cao,  S. Xu. ``Several new classes of Boolean functions with few Walsh transform values", arXiv:1506.04886v1
\bibitem{YG2006} N. Y. Yu and G. Gong, ``Construction of quadratic bent functions
in polynomial forms," IEEE Trans. Inf. Theory, vol. 52, no. 7,
pp. 3291-3299, Jul. 2006.

\end{thebibliography}
\end{document}